\documentclass[journal,10pt,doublecolumn]{IEEEtran}
\usepackage{amsmath, amssymb, amsfonts}
\usepackage{amsthm}
\usepackage{algorithm}
\usepackage{algorithmic}
\usepackage{color}
\usepackage{graphicx}
\usepackage{standalone}
\usepackage{filecontents}
\usepackage{tikz}
\usetikzlibrary{arrows,shadows,positioning}
\usetikzlibrary{decorations.pathreplacing,angles,quotes}
\usetikzlibrary{patterns}
\usetikzlibrary{chains,shapes.multipart}
\usetikzlibrary{shapes,calc}
\usetikzlibrary{automata,positioning}

\usepackage{caption}
\usepackage{epstopdf}
\usepackage{subfigure}
\usepackage{enumerate}
\usepackage{booktabs}
\usepackage{glossaries}

\usepackage{cite}

\def\BibTeX{{\rm B\kern-.05em{\sc i\kern-.025em b}\kern-.08em
		T\kern-.1667em\lower.7ex\hbox{E}\kern-.125emX}}

\DeclareMathOperator{\tr}{tr}

\DeclareMathOperator{\subjectto}{s.t.}

\DeclareMathOperator{\E}{\mathbb{E}}
\DeclareMathOperator{\Cov}{\mathbb{C}ov}
\DeclareMathOperator{\Prob}{\mathbb{P}}

\DeclareMathOperator{\argmin}{arg\min}

\newcommand{\norm}[1]{\left\lVert#1\right\rVert}

\newtheorem{theorem}{Theorem}
\newtheorem{proposition}{Proposition}

\newtheorem{remark}{Remark}

\algsetup{indent=0.7em}
\begin{document}

% paper title
% Titles are generally capitalized except for words such as a, an, and, as,
% at, but, by, for, in, nor, of, on, or, the, to and up, which are usually
% not capitalized unless they are the first or last word of the title.
% Linebreaks \\ can be used within to get better formatting as desired.
% Do not put math or special symbols in the title.
\title{Transmission Power Control for Remote State Estimation in Industrial Wireless Sensor Networks}

%\author{\IEEEauthorblockN{Samuele Zoppi, Touraj Soleymani, Markus Kl\"ugel, Mikhail Vilgelm, Sandra Hirche, and Wolfgang Kellerer}
%\IEEEauthorblockA{Chair of Communication Networks\\
%Technical University of Munich, Germany\\
%Email: \{samuele.zoppi, touraj, markus.kluegel, mikhail.vilgelm, sandra.hirche, kellerer.wolfgang\} @tum.de }
%}

%\author{\IEEEauthorblockN{
%		Samuele Zoppi\IEEEauthorrefmark{1},
%		Touraj Soleymani\IEEEauthorrefmark{2},
%		Markus Kl\"ugel\IEEEauthorrefmark{1},
%		Mikhail Vilgelm\IEEEauthorrefmark{1},
%		Sandra Hirche\IEEEauthorrefmark{2},
%		Wolfgang Kellerer\IEEEauthorrefmark{1}
%	}
%	\IEEEauthorblockA{
%		\IEEEauthorrefmark{1}Chair of Communication Networks, Department of Electrical
%		and Computer Engineering, TUM, Germany\\
%		\IEEEauthorrefmark{2}Chair of Information-Oriented Control, Department of Electrical 
%		and Computer Engineering, TUM, Germany\\
%		Email: \{samuele.zoppi, touraj, markus.kluegel, mikhail.vilgelm, sandra.hirche, kellerer.wolfgang\}\kern-1pt@tum.de		
%		\vspace{-2ex}
%	}
%}

\author{Samuele~Zoppi,~\IEEEmembership{Student Member,~IEEE,}
	Touraj~Soleymani,~\IEEEmembership{Student Member,~IEEE,}
	Markus~Kl\"ugel,~\IEEEmembership{Student Member,~IEEE,}
	Mikhail~Vilgelm,~\IEEEmembership{Student Member,~IEEE,}
	Sandra~Hirche,~\IEEEmembership{Senior~Member,~IEEE,}
	and~Wolfgang~Kellerer,~\IEEEmembership{Senior~Member,~IEEE}% <-this % stops a space
	\thanks{S. Zoppi, M. Kl\"ugel, M. Vilgelm are with the Chair of Communication Networks, Department of Electrical and Computer Engineering, TUM, Germany, e-mail: \{samuele.zoppi, markus.kluegel, mikhail.vilgelm, kellerer.wolfgang\}\kern-1pt@tum.de.}% <-this % stops a space
	\thanks{T. Soleymani and S. Hirche are with the Chair of Information-Oriented Control, Department of Electrical and Computer Engineering, TUM, Germany, e-mail: \{touraj, sandra.hirche\}\kern-1pt@tum.de.}% <-this % stops a space
	\thanks{Manuscript received July 16, 2019; revised mm dd, 2019.}}

% COPYRIGHT NOTICE
%\IEEEoverridecommandlockouts
%\IEEEpubid{\makebox[\columnwidth]{978-1-5386-0728-2/17/\$31.00~\copyright~2017 IEEE \hfill} \hspace{\columnsep}\makebox[\columnwidth]{ }}

%\makeatletter
%\def\ps@IEEEtitlepagestyle{
%  \def\@oddfoot{\mycopyrightnotice}
%  \def\@evenfoot{}
%}
%\def\mycopyrightnotice{
%  {\footnotesize
%  \begin{minipage}{\textwidth}
%  \centering
%  978-1-5386-0728-2/17/\$31.00~\copyright~2017 IEEE
%  \end{minipage}
%  }
%}

\maketitle
\thispagestyle{empty}

\begin{abstract}
Novel low-power wireless technologies and IoT applications open the door to the Industrial Internet of Things (IIoT).
In this new paradigm, Wireless Sensor Networks (WSNs) must fulfil, despite energy and transmission power limitations, the challenging communication requirements of advanced manufacturing processes and technologies. 
In industrial networks, this is possible thanks to the availability of network infrastructure and the presence of a network coordinator that efficiently allocates the available radio resources.
In this work, we consider a WSN that simultaneously transmits measurements of Networked Control Systems' (NCSs) dynamics to remote state estimators over a shared packet-erasure channel.
We develop a minimum transmission power control (TPC) policy for the coordination of the wireless medium by formulating an infinite horizon Markov decision process (MDP) optimization problem.
We compute the policy using an approximate value iteration algorithm and provide an extensive evaluation of its parameters in different interference scenarios and NCSs dynamics.
The evaluation results present a comprehensive characterization of the algorithm's performance, proving that it can flexibly adapt to arbitrary use cases.
\end{abstract}

\section{Introduction}\label{sec:introduction}
In the last few years, we have observed a paradigm shift for Wireless Sensor Networks (WSNs) from monitoring applications to industrial automation processes.
%With the advent of Industry 4.0~\cite{Hermann2016}, 
Industrial WSNs are foreseen to be integrated into advanced manufacturing techniques enabling the Industrial Internet of Things (IIoT). 
%WSNs deployed for   the IIoT 
They provide the necessary communication infrastructure for sensors and actuators to wirelessly operate in closed-loop control systems
called Networked Control Systems (NCSs). 

In state-of-the-art NCSs, state estimation is embedded in the design of closed-loop control policies~\cite{stoccontrol}.
% can be performed by independently performing the tasks of optimal estimation of the system's state and optimal control~\cite{stoccontrol}.
In particular, upon receiving the sensor's measurement, a remote state estimator calculates the Minimum Mean Square Error (MMSE) estimate of the system's state, which is then used by the controller to compute the actuation command.

Although wireless networks bring new capabilities and flexibility to NCSs, communication delays and packet dropouts highly affect the control performance and must be kept under control~\cite{Zhang2013}.
This is particularly critical for WSN devices that adopt low-power communication and are energy-constrained.
Therefore, in the IIoT, the available network resources have to be carefully coordinated to maximize the lifetime of the network while satisfying the communication requirements of the application.

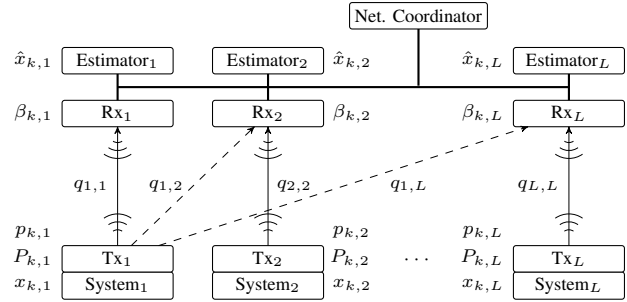
\begin{figure}
	  
 \begin{tikzpicture}[
  font = {\scriptsize},
  tx_line/.style={line width= 0.1pt,->,>=stealth',dashed},
  infra_line/.style={line width= 2pt,thick},
  box/.style={draw,rectangle,minimum height=\sy pt,minimum width=42 pt,anchor=center,fill=white,inner sep=2pt,outer sep=0pt,rounded corners=1pt},
  radiation/.style={{line width= 0.2pt,decorate,decoration={expanding waves,angle=35,segment length=3pt}}}]
  \def\dx{2}
  \def\by{45}
  \def\my{22.5}
  \def\rad{15}
  \def\sy{10}
  
  \newcommand{\ncsl}[3]{
  	\node[box] (p#3) at (#1,#2) {System$_#3$};
	\node[left=0pt of p#3] (pl#3) {$x_{k,#3}$};
  	
  	\node[box, above=0pt of p#3] (t#3) {Tx$_#3$};
  	\node[left=0pt of t#3] (tl#3) {$P_{k,#3}$};
  	\draw[radiation] ([xshift=0pt,yshift=3pt]t#3.north) -- ([xshift=0pt,yshift=\rad pt]t#3.north);
  	
  	\node[left=0 pt of t#3] [yshift=\sy](b#3) {$p_{k,#3}$};
  	%\node[left=0 pt of t#3] [yshift=\rad+\sy](b#3) {$q_{#3,#3}$};
  	
  	\node[box, above=\by pt of t#3] (r#3) {Rx$_#3$};
  	\node[left=0pt of r#3] (rl#3) {$\beta_{k,#3}$};
  	\draw[radiation] ([xshift=0pt,yshift=-3pt]r#3.south) -- ([xshift=0pt,yshift=-\rad pt]r#3.south);
  	
  	\node[box, above=\sy pt of r#3] (e#3) {Estimator$_#3$};
  	\node[left=0pt of e#3] (el#3) {$\hat{x}_{k,#3}$};
  	
  	\draw[infra_line] (r#3.north) to (e#3.south);}
  
    \newcommand{\ncsr}[3]{
  	\node[box] (p#3) at (#1,#2) {System$_#3$};
  	\node[right=0pt of p#3] (pl#3) {$x_{k,#3}$};
  	
  	\node[box, above=0pt of p#3] (t#3) {Tx$_#3$};
  	\node[right=-1pt of t#3] (tl#3) {$P_{k,#3}$};
  	\draw[radiation] ([xshift=0pt,yshift=3pt]t#3.north) -- ([xshift=0pt,yshift=\rad pt]t#3.north);
  	
  	\node[right=0 pt of t#3] [yshift=\sy](b#3) {$p_{k,#3}$};
  	%\node[left=0 pt of t#3] [yshift=\rad+\sy](b#3) {$q_{#3,#3}$};
  	
  	\node[box, above=\by pt of t#3] (r#3) {Rx$_#3$};
  	\node[right=0pt of r#3] (rl#3) {$\beta_{k,#3}$};
  	\draw[radiation] ([xshift=0pt,yshift=-3pt]r#3.south) -- ([xshift=0pt,yshift=-\rad pt]r#3.south);
  	
  	\node[box, above=\sy pt of r#3] (e#3) {Estimator$_#3$};
  	\node[right=0pt of e#3] (el#3) {$\hat{x}_{k,#3}$};
  	
  	\draw[infra_line] (r#3.north) to (e#3.south);}
  
	\ncsl{0}{0}{1}
	\ncsr{\dx}{0}{2}
	%\ncs{2*\dx}{0}{3}
	\ncsl{3*\dx}{0}{L}

	\draw[tx_line,solid](t1) to node[left] {$q_{1,1}$} (r1);
	\draw[tx_line](t1) to node[left] {$q_{1,2}$} (r2);
	\draw[tx_line](t1) to node[xshift=25pt] {$q_{1,L}$} (rL);
	\draw[tx_line,solid](t2) to node[right] {$q_{2,2}$} (r2);
	\draw[tx_line,solid](tL) to node[left] {$q_{L,L}$} (rL);
	
	\draw[infra_line] ([xshift=0pt,yshift=5pt]r1.north) to ([xshift=0pt,yshift=5pt]rL.north);
	\node (dots) at (2*\dx,\sy pt) {$\cdots$};
	\node (coord-net) at ([xshift=\dx cm,yshift=2pt]r2.north) {};
	\node[box] (coord) at ([yshift=30 pt]coord-net) {Net. Coordinator};
	\draw[infra_line] (coord-net) to (coord);

\end{tikzpicture}
	%\vspace{0.1cm}
	\caption{System model of an industrial WSN deployed for the remote state estimation of multiple dynamical systems.}
	\label{fig:system-architecture}
\end{figure}

In this work, we consider an industrial WSN deployed to convey measurements of multiple independent NCSs' dynamics to remote state estimators as in Fig.~\ref{fig:system-architecture}.
We study the impact of packet dropouts arising from the simultaneous transmission of sensor measurements over a shared communication channel.
For successful simultaneous communication, the transmission powers of the sensors need to be coordinated to control the interference levels in the wireless medium.
In industrial WSNs this is possible thanks to the availability of a network infrastructure that allows the centralized coordination of the communication resources~\cite{Gungor2009}.
%This way, sensor measurements can timely reach the remote estimators and transmission powers can be optimally coordinated.
The objective of the network coordinator is to control the transmission powers of the sensors to ensure the correct operation of the remote state estimators while minimizing the energy expenditure of the network.

For this reason, in this work, we propose a method to determine the minimum transmission power control (TPC) policy for a WSN deployed for remote state estimation of NCSs' dynamics.
We achieve this by (i) formulating an infinite horizon Markov decision process (MDP) optimization problem,
%that combines the network's transmission powers and estimation error covariances
(ii) using an approximate value iteration algorithm to solve it, and (iii) exhaustively evaluating it in different interference scenarios and configurations.

The remainder of this paper is structured as follows. 
In Sec.~\ref{sec:related-work} the related work is discussed.
Sec.~\ref{sec:tpc-problem} details the propagation model of the WSN and the Foschini-Miljanic TPC algorithm, while Sec.~\ref{sec:remote-state-estimation} describes dynamical systems' model and the remote state estimation procedure.
Furthermore, Sec.~\ref{sec:opt-problem} formulates the proposed optimization problem, Sec.~\ref{sec:opt-algorithm} describes its solution via an approximate algorithm, and Sec.~\ref{sec:evaluation} presents its comprehensive evaluation.
Finally, in Sec.~\ref{sec:conclusions}, conclusions are drawn and future work is discussed.

\subsection{Related Work}\label{sec:related-work}
%Control: tpc for remote state estimation 
The investigation of optimal state estimation of dynamical systems with intermittent measurements started approximately a decade ago~\cite{Sinopoli2004}.
Thanks to the suitability of WSNs for IIoT applications, packet dropouts as the cause of intermittent measurement reporting have been largely investigated.
Initially, the trade-off between transmission power and packet loss for remote state estimation of a single-sensor has been studied in %\cite{Li2013a}
\cite{Ren2018,Wu2015,Wu2017,Soleymani2018,Li2013,Li2017}.
Solutions are presented for a noiseless sensor~\cite{Ren2018}, and for noisy measurements where the transmission power policy is computed offline~\cite{Wu2015,Wu2017,Soleymani2018}, online~\cite{Li2013}, and in combination with an energy harvester~\cite{Li2017}.

Additional research studies have investigated the scenario of multiple sensors sharing the same communication channel~\cite{Li2014,Weerakkody2016,Wu2018,Zhang2017,Li2019a,Li2018,Li2019}.
In particular, Li, Weerakkody, and Wu et al.~\cite{Li2014,Weerakkody2016,Wu2018} studied the optimal scheduling of sensor transmissions in TDMA and CSMA medium access schemes for event-based and periodic measurements.
When the sensors simultaneously access the wireless medium, transmission power control schemes have been investigated for independent estimators %\cite{Li2014}
\cite{Zhang2017,Li2019a} and in the context of sensor fusion~\cite{Li2018,Li2019}.
Existing TPC solutions for multiple independent estimators investigate distributed game-theoretical methods, that rely on partial information, require convergence time, and are tailored to legacy uncoordinated WSNs.
%Moreover, they all neglect the impact of interference originated by different topologies on the optimal coordination of transmission powers. 
%For this reason, we investigate the centralized coordination of transmission powers in WSNs with network infrastructure and determine its performance for realistic IIoT topologies.
The investigation of optimal allocation of transmission powers in WSNs with network infrastructure still remains open.
In our IIoT scenario, we make use of a centralized network controller and complete channel information to design an optimal offline TPC policy that does not require online convergence.

% Communication: utility based tpc in wireless network 
%Utility-based power control SotA... add book from Utschick - Boche -Stanzack
%Also~\cite{4278411,xiao2003utility}
On the other hand, communication research devoted a remarkable effort over the last decades to optimally coordinate transmission powers, in particular in the context of cellular networks \cite{Chiang2007,Douros2011}.
Results show that the TPC problem can be both formulated as centralized or distributed problem thanks to the definition of feasible communication links according to the coupling introduced by the signal-to-interference-and-noise ratio (SINR)~\cite{Foschini1993}.
In this direction, research works have investigated the problems of scheduling~\cite{Borbash2006} and energy efficiency in multi-hop ad-hoc networks~\cite{Cruz2004}.

Most of existing TPC algorithms for wireless ad-hoc networks assume a lack of network infrastructure and diverse objectives for the sensors in the network~\cite{JianweiHuangMemberIEEERandallA.BerryMemberIEEEandMichaelL.HonigFellow2006,St.Jean2005,Long2007,Sengupta2010}.
In particular, cooperative~\cite{JianweiHuangMemberIEEERandallA.BerryMemberIEEEandMichaelL.HonigFellow2006} and non-cooperative~\cite{St.Jean2005,Long2007,Sengupta2010} game-theoretical solutions have been developed tackling traditional communication objectives such as energy efficiency~\cite{Long2007}, throughput~\cite{St.Jean2005}, or an arbitrary utility~\cite{JianweiHuangMemberIEEERandallA.BerryMemberIEEEandMichaelL.HonigFellow2006,Sengupta2010}.
A centralized TPC strategy for ad-hoc networks has been investigated that jointly optimizes throughput, delay and power consumption~\cite{Li2007}.

Although TPC techniques have been extensively studied in communication networks, no research work tackles the problem of coordinating transmission powers for multiple remote state estimators of NCSs' dynamics.
Existing work focuses on standard communication metrics or arbitrary convex utilities that cannot be related to the problem of remote state estimation.
Furthermore, existing ad-hoc network solutions do not consider the IIoT capability of providing network infrastructure, that opens new opportunities in the coordination of WSNs.

\section{System Model and Background}\label{sec:problem-formulation}
We consider an Industrial IoT scenario where $L$ LTI dynamical systems operate in an indoor environment as represented in Fig.~\ref{fig:system-architecture}.
Every system is equipped with a wireless sensor that samples its state and transmits it to a remote state estimator 
using a one-step delay packet-erasure channel with acknowledgement.
The WSN consists of $2L$ sensors adopting the IEEE Std. 802.15.4 physical layer~\cite{802154-2006} transmitting on the same channel according to the TDMA medium access control IEEE Std. 802.15.4e\cite{802154-2012}. % that adopts TDMA in combination with frequency hopping.
Upon receiving the sensor's measurement, a remote state estimator calculates the Minimum Mean Square Error (MMSE) estimates of the systems, which are then used by the controller to compute the actuation commands.

This section provides the system model and background of the WSN and the remote state estimation of dynamical systems.
Sec.~\ref{sec:tpc-problem} defines the propagation model of the WSN and the Foschini-Miljanic TPC algorithm, while Sec.~\ref{sec:remote-state-estimation} details the dynamical systems' model and the MMSE remote state estimation procedure. 
Tab.~\ref{tab:list-of-symbols} summarizes the overall system model parameters.

\begin{table}
	\small
	\begin{tabular}{p{1.2cm} p{6.6cm}}
		\toprule
		Symbol & Description \\
		\midrule
		$q_{\ell,m}$& Channel coefficient.\\
		$\zeta_{\ell,m}$& Path loss coefficient.\\
		$d_{\ell,m}$& Relative communication distance.\\
		$\eta$& Path loss decay.\\
		$f$& Transmission frequency.\\
		$\mu_\xi$& Average multi-path fading.\\
		$\sigma^2$& Variance of the multi-path fading's logarithm.\\
		$\gamma_{k,\ell}$& Signal-to-interference-and-noise ratio.\\
		$p_{k,\ell}$& Transmission power.\\
		$n_{k,\ell}$& Power of the additive white Gaussian noise.\\
		$\beta_{k,\ell}$& Random packet-erasure process.\\
		$\kappa_{k,\ell}$& Average Packet Success Ratio.\\
		$W$& Packet length.\\
		\midrule
		$L$& Total number of dynamical systems.\\
		$x_{k,\ell}$& Instantaneous system's state.\\
		$F_{\ell}$& State matrix. \\
		$v_{k,\ell}$& State noise.\\
		$y_{k,\ell}$& System's observation.\\ 
		$H_{\ell}$& Output matrix.\\
		$w_{k,\ell}$& Measurement noise.\\
		$R_{1,\ell}, R_{2,\ell}$& Cov. matrices of the state and measurement noises.\\
		$\mathcal{I}_{k,\ell}$& Estimator's available information.\\
		$\Phi_{\ell}$& System's estimation distortion.\\
		$\hat{x}_{k,\ell}$& Estimated system's state.\\
		$P_{k,\ell}$& Estimation error covariance.\\
		$K_{k,\ell}$& Kalman Filter's gain.\\
		\bottomrule
	\end{tabular}
	\caption{List of system model parameters.}
	\label{tab:list-of-symbols}
\end{table}

\subsection{Wireless Propagation and Transmission Power Control}\label{sec:tpc-problem}
The wireless propagation between the $m$-th transmitter and the $\ell$-th receiver is modelled by the channel coefficient $q_{\ell,m}$, which captures the average fraction of received power after path loss and multi-path fading. 
The first is modelled by a logarithmic path loss model, while the second by a log-normal random variable suitable for low-power indoor WSN communication~\cite{akyildiz2010wireless,zamalloa2007analysis}
\begin{align}
q_{\ell,m}&=\zeta_{\ell,m}\, {\mu_\xi}^{-1}\mbox{, }\quad q_{\ell,m}<1\mbox{,} \label{eq:channel-coefficient}\\
\zeta_{\ell,m} &= \left(\frac{c_0}{4\pi f d_0}\right)^2 \left(\frac{d_0}{d_{\ell,m}}\right)^\eta \mbox{,}\label{eq:path-loss}\\ 
\mu_\xi &= \E\left[\xi\right] = e^{\sigma^2 / 2} \mbox{, }\quad \ln\left(\xi\right)\sim \mathcal{N}\left(0, \sigma^2\right) .\label{eq:shadowing}
\end{align}
Where, $\mu_\xi$ is the expected value of the log-normal stochastic variable $\xi$ with parameter $\sigma^2$. 
The path loss component $\zeta_{\ell,m}$ depends on the speed of light $c_0$, the operating frequency $f$, the path loss decay $\eta$, and the relative distance between the transmitter and the receiver $d_{m,\ell}$ normalized by a reference distance $d_0$.

All sensors share the same wireless channel for the transmission of sensor values to their respective remote state estimators.
At every time step $k$, the communication quality of the $\ell$-th link is affected by the simultaneous transmissions of the other $L-1$ sensors and is described by the signal-to-interference-and-noise ratio (SINR)
\begin{equation}\label{eq:sinr}
\gamma_{k,\ell} = \frac{p_{k,\ell} \, q_{\ell,\ell}}{\sum_{m\neq \ell}p_{k,m} \, q_{m,\ell} + n_{k,\ell}},
\end{equation}
where $p_{k,\ell}$ is the transmission power of the $\ell$-th transmitter and $n_{k,\ell}$ the power of the additive white Gaussian noise. %, which is characteristic of every receiver $l$.
%In our system, $\gamma_{k,m}$ indicates the \textit{average} SINR experienced by the receiver of link $l$.
A specific SINR value is mapped to a Packet Success Ratio (PSR) value according to the OQPSK modulation with DSSS (IEEE Std. 802.15.4~\cite{802154-2006}) and CRC recovery mechanism.
Therefore, the average PSR of a packet of $W$ bits transmitted over the $\ell$-th link is given by
\begin{align}\label{eq:psr-vs-sinr}
\kappa_{k,\ell} = f\left(\gamma_{k,\ell}\right) = \left[1 - Q\left(4\sqrt{\gamma_{k,\ell}} \right)\right]^{W},
\end{align}
where $Q(\cdot)$ denotes the standard Gaussian error function.

The packet-erasure process of the $\ell$-th lossy communication channel is modelled by a Bernoulli random variable $\beta_{k,\ell}$ with mean equal to the PSR $\kappa_{k,\ell}$, i.e,
\begin{align}
\Pr\left[\beta_{k,\ell} = b \right] &= \left\{
\begin{array}{l l}
1 & \ \text{with probability } \kappa_{k,\ell}, \\
0 & \ \text{otherwise}.
\end{array} \right.\label{eq:packet-erasure-process}
\end{align}

The simultaneous transmission of multiple sensors on the same channel generates interference that reduces the SINR and creates packet loss. %affects the SINR.
In order to achieve the desired quality for all the communication links, the transmission powers need to be adjusted accordingly. 
The problem of optimal allocation of transmission powers for multiple interfering transmitters is tackled by the Foschini-Miljanic algorithm~\cite{Foschini1993}.
Given the network PSR requirements and wireless propagation parameters, the algorithm computes the minimum transmission powers by solving the system of equations arising from~\eqref{eq:sinr} which describes the SINR coupling of the network.

In this work, we apply the Foschini-Miljanic algorithm~\cite{Foschini1993} to coordinate the simultaneous transmission of measurements.
Given the PSR requirements of the network $\vec{\kappa}_k = \left[\kappa_{k,1}, \dots, \kappa_{k,L}\right]$, it is possible to calculate the corresponding SINR vector $\vec{\gamma}_k^{\,\kappa} = \left[\gamma_{k,1}^\kappa, \dots, \gamma_{k,L}^\kappa\right]^T$ using ~\eqref{eq:psr-vs-sinr}
\begin{equation}\label{eq:snr-from-psr-requirement}
\gamma_{k,i}^\kappa = f^{-1}\left(\kappa_{k,i}\right).
\end{equation}

The vector of the minimum transmission powers satisfying the SINR requirements of the network is~\cite{Chiang2007}
\begin{align}
	\vec{p}_{k}^{\,\kappa} &= g\left(\vec{\gamma}_k^{\,\kappa}\right) = {\left( I_L - D\left(\vec{\gamma}_k^{\,\kappa}\right)T \right)}^{-1} \vec{u}_k,\label{eq:foschini}\\
	\vec{\gamma}_k^{\,\kappa} &\in \Lambda_\gamma \triangleq \left\{ \vec{\gamma}_k^{\,\kappa} \ge 0: 0 \leq g\left(\vec{\gamma}_k^{\,\kappa}\right)\leq p_\text{max} \right\}. \label{eq:feasibility-sinr}
\end{align}
Where $I_L$ is the identity matrix, $D\left(\vec{\gamma}^\kappa_k\right)$ is a diagonal matrix with SINR requirements $\vec{\gamma}^\kappa_k$ on the main diagonal, $T$ is the normalized-gain matrix~\eqref{eq:normalized-gain-matrix}, $\vec{u}_k$ the normalized interference vector~\eqref{eq:normalized-interference-vector}, and $\Lambda_\gamma$ is  the feasibility region, i.e. the space of feasible SINR requirements which depends on the maximum transmission power $p_\text{max}$.
\begin{align}
T &= 
\begin{bmatrix}
0 & q_{1,2}/q_{2,2} & \dots & q_{1,L}/q_{L,L} \\
q_{2,1}/q_{1,1} & 0 & \dots & q_{2,L}/q_{L,L} \\
\dots  & \dots & \dots & \dots \\
q_{L,1}/q_{1,1} & q_{L,2}/q_{2,2} & \dots & 0
\end{bmatrix} ,\label{eq:normalized-gain-matrix}\\
\vec{u}_k &= \left[ \frac{n_{k,1}\gamma_1^\kappa}{q_{1,1}},\frac{n_{k,2}\gamma_2^\kappa }{q_{2,2}}, \dots, \frac{n_{k,L}\gamma_L^\kappa}{q_{L,L}} \right]^\top\label{eq:normalized-interference-vector},\\
\vec{n}_k &= \left[ n_{k,1}, n_{k,2}, \dots, n_{k,L}\right]^\top.
\end{align}

\begin{figure}[t!]
	\includegraphics{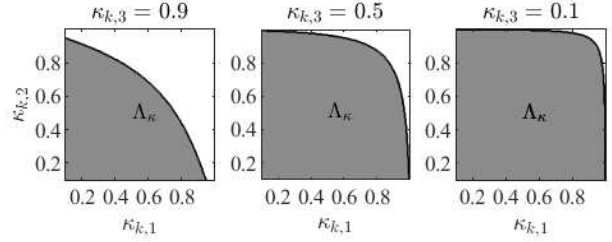}
	\caption{Exemplary PSR feasibility region $\Lambda_\kappa$ of 3 sensors with equal distances, for sensors $1$ and $2$ ($\kappa_{k,1},\kappa_{k,2}$) as the PSR req. of sensor $3$ ($\kappa_{k,3}$) decreases (left to right: 0.1, 0.5, 0.9).}
	\label{fig:psr-feasibility-region}
\end{figure}

Combining~\eqref{eq:psr-vs-sinr} and~\eqref{eq:foschini}, the feasibility region $\Lambda_\gamma$ can be expressed in terms of PSR requirements as
\begin{align}
\vec{\kappa}_k \in \Lambda_\kappa &= \left\{ \vec{\kappa}_k\ge 0: 0 \leq\vec{\Psi}(\vec{\kappa}_k) \leq p_{max} \right\},\\
\vec{\Psi}(\vec{\kappa}_k) &= g\left( f^{-1} \left(\vec{\kappa}_k\right)\right).
\end{align}
Fig.~\ref{fig:psr-feasibility-region} shows the PSR feasibility region  $\Lambda_\kappa$ for an exemplary case of $3$ sensors. 
Each diagram shows the pair of feasible PSR requirements of sensor $1$ and $2$ for different values of $\kappa_{k,3}$, i.e. $0.1,0.5,$ and $0.9$.
In this scenario, the distances between all the transmitters and the receivers are the same.
%, this scenario arises whenever the distance between the different transmitters is much smaller than the one between transmitter and receiver.
This way, the bi-dimensional feasibility regions are symmetric and the feasibility regions of sensors $2,3$ and $1,3$ are equivalent to the ones shown in Fig.~\ref{fig:psr-feasibility-region}.
We notice that, due to the coupling of the channels, whenever the PSR requirement of a sensor is more stringent ($\kappa_{k,3}=0.9$), the feasible PSR requirements of the other users are restricted.
In our study, the feasibility region varies according the position the propagation parameters of all the sensors and represents the set of possible power and PSR allocations of the WSN.
%we assume a static or slowly varying 

\subsection{Dynamical Systems and Remote State Estimation}\label{sec:remote-state-estimation}
The discrete-time dynamics of the $\ell$-th system are generated by the following linear state equation
\begin{align}\label{eq:system-dynamics}
x_{k+1,\ell} &= F_\ell \, x_{k,\ell} +  w_{k,\ell},
\end{align}
for time $k \in \mathbb{N}_+$ and with initial condition $x_{0,\ell}$ where $x_{k,\ell} \in \mathbb{R}^n$ is the state of the system, $F_\ell$ is the state matrix, $w_{k,\ell} \in \mathbb{R}^n$ is a white Gaussian noise process with zero mean and covariance $R_{1,\ell}$ where $R_{1,\ell} \succ 0, \forall \ell$. 
At each time step, the state $x_{k,\ell}$ is observed by a sensor whose measurement is
\begin{align}\label{eq:system-observations}
y_{k,\ell} &= H_\ell \, x_{k,\ell} + v_{k,\ell},
\end{align}
where $y_{k,\ell} \in \mathbb{R}^p$ is the measurement of the system, $H_\ell$ is the output matrix, and $v_{k,\ell} \in \mathbb{R}^p$ is a white Gaussian noise process with zero mean and covariance $R_{2,\ell}$ where $R_{2,\ell} \succ 0, \forall \ell$. 
It is assumed that the initial state $x_{0,\ell}$ is a Gaussian vector with mean $m_{0,\ell}$ and covariance $R_{0,\ell}$, that $x_{0,\ell}$, $w_{k,\ell}$, $v_{k,\ell}$ are mutually independent, and that $\left(F_\ell,H_\ell\right)$ is observable $\forall \ell$.

The quality of estimation depends on the information available at the estimator.
At time $k$, for the $\ell$-th control system, we assume that the information available at the estimator is specified by
\begin{align}\label{eq:available-information-at-controller}
\mathcal{I}_{k,\ell} = \left\{ \beta_{0:k-1,\ell}, {y}_{0:k-1,\ell} \right\}.
\end{align}

We measure the estimation distortion of the $\ell$-th system by the quadratic function
\begin{equation}\label{eq:estimation-distortion-system-l}
\Phi_{\ell} = \E \left[ \sum_{k=0}^{N} \left\| x_{k,\ell} - \hat{x}_{k,\ell} \right\|_{\Theta_{k,\ell}}^2 \right],
\end{equation}
where $\Theta_{k,\ell} \succ 0$ is a weighting matrix and $\hat{x}_{k,\ell}$ is the best state estimate given the information set $\mathcal{I}_{k,\ell}$.

As shown in~\cite{Soleymani2018}, the conditional expected value of the state %with the following dynamics 
is the minimizer of the mean square error for the system in~\eqref{eq:system-dynamics} and~\eqref{eq:system-observations} over the one-step delay packet-erasure channel with arrival process specified by~\eqref{eq:packet-erasure-process}.
Hence, we can obtain
\begin{subequations}\label{eqs:kalman}
	\begin{align}
	\hat{x}_{k+1,\ell} &= F_\ell \hat{x}_{k,\ell} + \beta_{k,\ell}  K_{k,\ell} (y_{k,\ell}\!-\!H_\ell \hat{x}_{k,\ell}),\\
	P_{k+1,\ell} &= F_\ell P_{k,\ell} F_\ell^\top + R_{1,\ell} - \beta_{k,\ell} K_{k\,\ell} H_\ell P_{k,\ell} F_\ell^\top,
	\end{align}
	where $\hat{x}_{k,\ell} = \E[{x}_{k,\ell} | \mathcal{I}_{k,\ell}]$, $P_{k,\ell} = \Cov[{x}_{k,\ell} | \mathcal{I}_{k,\ell}]$, and
	\begin{align}
	K_{k,\ell} &= F_\ell P_{k,\ell} H_\ell^\top \left(H_\ell P_{k,\ell} H_\ell^\top + R_{2,\ell}\right)^{-1},
	\end{align}
	with initial conditions $\hat{x}_{0,\ell} = m_{0,\ell}$ and $P_{0,\ell} = R_{0,\ell}$, $\forall \ell$.
\end{subequations}
\section{Optimal Transmission Power Allocation Problem}\label{sec:opt-problem}

%For every remote state estimation, we are interested in the optimal allocation of transmission powers subject to a desired level of estimation distortion $c_\ell$,
%In this Section, we formulate the problem of optimal joint transmission power allocation subject to a desired level of estimation distortion.
We are interested in finding the minimum transmission powers required for satisfying the desired level of estimation distortion in each remote state estimator.

Let $\vec{p}_k=\left[ p_{k,1}, \dots, p_{k,L}\right]^\top \in \mathcal{P} \, \forall k$ represent all the possible transmission policies and $\mathcal{P} = \left\lbrace \vec{p}_k : 0 \le \vec{p}_k \le p_\text{max}  \right\rbrace$ be the set of feasible transmission powers.
%Using~\eqref{eq:estimation-distortion-system-l}, we can measure and combine the global estimation distortion with the network transmission power allocation.
The target is to find $\vec{p}_k \in \mathcal{P}$ that minimizes
\begin{align}\label{eq:opt-relaxed}
\inf_{\vec{p}_k \in \mathcal{P}} \lim_{N\to\infty} \E \left[ \sum_{\ell=1}^L \sum_{k=0}^N  \alpha^k p_{k,\ell} + \alpha^k\lambda_\ell \norm{ x_{k,\ell}  -\hat{x}_{k,\ell}}_{\Theta_{k,\ell}}^2 \right],
\end{align}
where $\lambda_\ell$ regulates the trade-off between the estimation distortion of a single remote state estimator and the transmission power of the sensor, and $\alpha \in \left(0,1\right)$ is the \emph{discount factor} which weights the relative contribution of the costs in the short-term and long-term
future. 

We concentrate on the estimation error covariance $P_{k,\ell}$ and use the identity $\E \left[\norm{ x_{k,\ell}  -\hat{x}_{k,\ell}}_{\Theta_{k,\ell}}^2\right]= \E \left[\tr(\Theta_{k,\ell} {P}_{k,\ell})\right]$ to reformulate the estimation error~\cite{Soleymani2018}.
Hence, the optimal transmit power schedule is obtained by solving the following optimization problem
\begin{alignat}{2}\label{eq:opt_prob1}
&\min_{\vec{p}_k \in \mathcal{P}}  && \lim_{N\to\infty} \E \left[ \sum_{\ell=1}^{L} \sum_{k=0}^{N} \alpha^k p_{k,\ell} + \alpha^k \lambda_\ell \tr(\Theta_{k,\ell} {P}_{k,\ell}) \right], \notag\\
& \subjectto && P_{k+1,\ell} = F_\ell P_{k,\ell} F_\ell^\top + R_{1,\ell} - \beta_{k,\ell} K_{k,\ell} H_\ell P_{k,\ell} F_{\ell}^\top,\notag\\
&  &&  \vec{p}_k= \vec{\Psi}(\vec{\kappa}_k),
\end{alignat}
with initial conditions $P_{0,\ell}=R_{0,\ell}, \forall \ell$.

% We solve the optimization problem of Eq.~ \eqref{eq:opt_prob1} by means of a Markov decision process (MDP). 
% The PSR feasibility region $\Lambda_{\kappa}$ of the transmission power control problem of Sec.~\ref{sec:tpc-problem} represents the \emph{action space} of the MDP.
% The objective of the MDP is to find a steady-state PSR policy that corresponds to a minimum transmission power allocation of the Foschini-Miljanic algorithm.
Note that it is possible to separate the design of the optimal estimation distortion from the allocation of the optimal transmission powers.
Although the systems are coupled via the shared communication channel, the optimal allocation of powers given the coupling is provided by the Foschini-Miljanic algorithm for feasible PSR requirements.
Therefore, the optimization problem can be equivalently expressed in terms of optimal PSR requirement.

\begin{theorem}
\label{th:separation}
There is a separation between designs of the optimal estimate and the optimal power schedule.
\end{theorem}
\begin{proof}
The proof follows the above derivations. The optimal estimate is obtained by the recursive filter given in~\eqref{eqs:kalman} and the optimal power schedule is obtained by solving the optimization problem in~\eqref{eq:opt_prob1}.
\end{proof}

By using the separation property, we can obtain an equivalent optimization problem in terms of $\vec{\kappa}_k$, where 
\begin{IEEEeqnarray}{l}
\min_{\vec{\kappa}_k \in \Lambda_\kappa} \lim_{N\to\infty} \E \left[ \sum_{\ell=1}^{L} \sum_{k=0}^{N} \alpha^k \vec{\Psi}_\ell \left(\vec{\kappa}_k\right) + \alpha^k \lambda_\ell \tr(\Theta_{k,\ell} {P}_{k,\ell})\right], \nonumber\\
\subjectto\, P_{k+1,\ell} = F_\ell P_{k,\ell} F_\ell^\top + R_{1,\ell} - \beta_{k,\ell} K_{k,\ell} H_\ell P_{k,\ell} F_{\ell}^\top.
\end{IEEEeqnarray}

The problem above can be viewed as an infinite horizon Markov decision process (MDP), where the goal is to find an optimal steady-state PSR policy that corresponds to the minimum transmission power allocation of the Foschini-Miljanic algorithm. The MDP has a state space of all possible covariance combinations and an action space determined by the PSR feasibility region $\Lambda_{\kappa}$ of the transmission power control problem of Sec.~\ref{sec:tpc-problem}.

At time $k$, we denote the \emph{action} of the system as $\vec{\kappa_k}$ and its \emph{state} as $S_k\triangleq\left[P_{k,1},\dots,P_{k,L}\right]^\top$, where $S_0$ denotes the initial state. 

The stage cost is defined as
\begin{align}
\rho \left( S_{k}, \vec{\kappa}_k\right) &\triangleq \sum_{\ell=1}^{L} \vec{\Psi}_\ell(\vec{\kappa}_k) + \lambda_\ell \tr(\Theta_{k,\ell} {P}_{k,\ell}).
\end{align}

Given the random outcome vector 
$\vec{\beta}_k=\left[\beta_{k,1},\dots,\beta_{k,L}\right]^\top$, we can define a covariance transition function $\phi_l: P_{k,\ell}\to P_{k+1,\ell}$ and the state transition function $\Phi: S_k\rightarrow S_{k+1}$ as follows
\begin{align}
\phi_\ell\left(P_{k},
\beta_k\right) &\triangleq F_\ell P_{k,\ell} F_\ell^\top + R_{1,\ell} - \beta_{k,\ell} K_{k,\ell} H_\ell P_{k,\ell} F_{\ell}^\top,\\
\Phi\left( S_{k},
\vec{\beta}_k\right) &\triangleq \left[\phi_1(P_{k,1},\beta_{k,1}),\dots,\phi_L(P_{k,L},\beta_{k,L})\right]\top.
\end{align}

The optimization problem in~\eqref{eq:opt_prob1} can be concisely rewritten as
\begin{subequations}
\label{eq:final_opt}
\begin{align}
\min_{\vec{\kappa}_k \in \Lambda_k}& \lim_{N\to\infty} \E \left[ \sum_{k=0}^{N}\alpha^k \rho \left( S_{k}, \vec{\kappa}_k \right) \right]\\ 
\text{ s.t. }& S_{k+1} = \Phi \left( S_{k}, \vec{\beta}_k\right),\\
&\vec{\kappa}_{k} = \E\left[\vec{\beta}_{k}\right],
\end{align}
\end{subequations}
with initial condition $S_{0}= \left[R_{0,1},\dots,R_{0,L}\right]^\top$ and $\alpha \in \left(0,1\right)$.
\section{Approximate Value Iteration Algorithm}\label{sec:opt-algorithm}
\begin{figure*}[t]
	\includegraphics{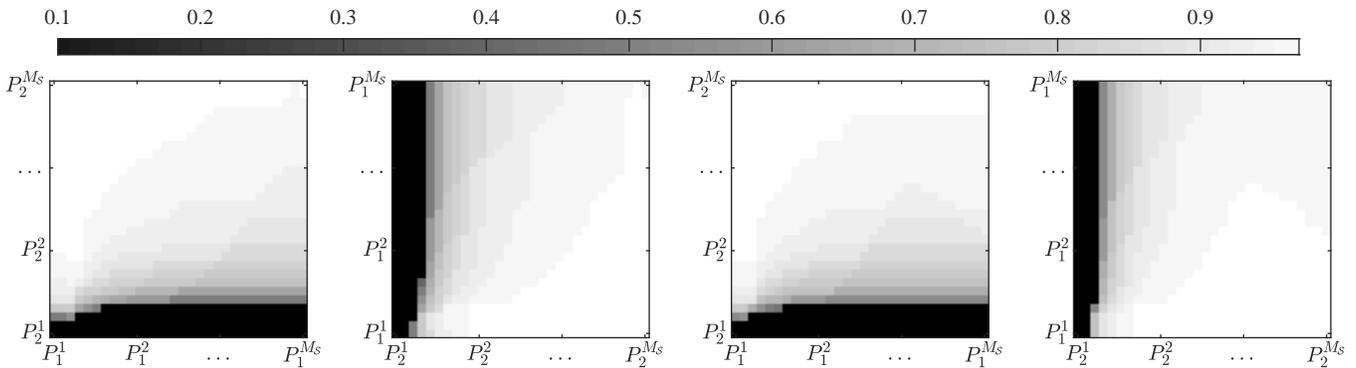}
	%	\caption{PSR Schedule User 1.}
	%	\label{fig:psr-schedule-u1}
	%\end{figure*}
	
	%\begin{figure*}[t]
	%	\centering
	%	\begin{minipage}{.5\textwidth}
	%		\centering
	%		\includegraphics{./matlab/output/policies/psr_policy_policies_users_2_asize_30_d1cm_1000_d2cm_1000_psize_30_alpha_7500_lambda_2286_epsilon_500}
	%		%\captionof{figure}{A figure}
	%		%\label{fig:test1}
	%	\end{minipage}%
	%	\begin{minipage}{.5\textwidth}
	%		\centering
	%		\includegraphics{./matlab/output/policies/psr_policy_policies_het_users_2_asize_30_d1cm_1000_d2cm_1000_psize_30_alpha_7500_lambda_2286_epsilon_500}
	%		%\captionof{figure}{Another figure}
	%		%\label{fig:test2}
	%	\end{minipage}
	\caption{Visualization of exemplary PSR policies of two sensors ($L=2$) with equal communication distances ($10$ m) and discretization $M_\mathcal{S}=M_{\bar{\Lambda}}=P_1^{M_\mathcal{S}}=P_2^{M_\mathcal{S}}=30$. 
	The left pair of plots shows the PSR policies of two homogeneous systems $F_1=F_2=1.01$, while the right pair shows those for heterogeneous systems $F_1=1.01, F_2=1.2$.}
	\label{fig:psr-schedules}
\end{figure*}

\begin{table}
	\small
	\begin{tabular}{p{1.2cm} p{6.6cm}}
	\toprule
	Symbol & Description \\
	\midrule
	$\bar{\Lambda}, M_{\bar{\Lambda}}$&  Discretized action space and its cardinality.\\
	$\mathcal{S}, M_\mathcal{S}$&  Discretized state space and its cardinality.\\
	$\rho$& Stage cost function.\\
	$J$& Target update function.\\
	$\phi_\ell$&  Estimation error  covariance function.\\
	$\Phi$&  State transition function.\\
	$\lambda_\ell$& Estimation error to transmission power trade-off.\\
	$\alpha$& Discount factor.\\
	$\epsilon$& Algorithm's precision.\\
	\bottomrule
	\end{tabular}
	\caption{List of the value iteration algorithm's parameters.}
	\label{tab:list-of-algorithm-symbols}
\end{table}

In this section, we develop an approximate algorithm to solve the optimization problem in~\eqref{eq:final_opt}.

Assuming that the distribution of PSR is stationary within the horizon $N$, and hence its expectation and PSR feasibility region are constant within a sufficiently large horizon, the problem in~\eqref{eq:final_opt} is an infinite horizon discounted cost problem defined on an MDP. 
This problem can be solved by dynamic programming algorithms such as value iteration or policy iteration. 
These algorithms are guaranteed to converge to the globally optimal stationary policy~\cite{Bertsekas_DP}.

However, both state space and action space are continuous. 
Therefore, we need to apply an approximate algorithm with discretized state and action space\footnote{In practice, possible allocated powers and thus action space is discrete due to implementation or standard specifications~\cite{802154-2006}.}.
At every time step (we further omit the time-step index for readability), we define a discretized system state space $\mathcal{S}$ where the covariance space of each system is discretized with $M_{\mathcal{S}}$ levels
\begin{align}
&S\in \mathcal{S},\, S=\left[P_1,\dots,P_L\right]^\top \text{ where}\nonumber\\
&P_\ell\in \{P_{\ell}^{1},\dots,P_{\ell}^{M_\mathcal{S}}\}\quad\forall \ell\in\{1,\dots,L\}.\nonumber
\end{align}

Similarly, we define a discretized action space $\bar{\Lambda}$ as all combinations of PSRs belonging to the discretized feasibility region of cardinality $|\bar{\Lambda}|\triangleq M_{\bar{\Lambda}}$
\begin{align}
&\vec{\kappa}\in \bar{\Lambda},\, \vec{\kappa}=\left[\kappa_1,\dots,\kappa_L\right]^\top \text{ where}\nonumber\\
&\kappa_\ell\in \{\kappa_{\ell}^{1},\dots,\kappa_{\ell}^{M_{\bar{\Lambda}}}\}\quad\forall \ell\in\{1,\dots,L\}.\nonumber
\end{align}

%\samuele{I would distinguish between the $\Lambda_\kappa$ (continuous) and the action space. What about introducing $\mathcal{A}$ with cardinality $N$? We can use $T$ for the horizon in the optimization problem.}
Approximate value iteration algorithm recursively computes the value function $J(S)$ for every state $S\in\mathcal{S}$, using a target update step based on the Bellman optimality equation. We define the target update via the following proposition. 

\begin{proposition}
The target update for value iteration is expressed as
\begin{equation}
J(S)=\min_{\vec{\kappa}\in\bar{\Lambda}}\left\lbrace\rho\left(S,\vec{\kappa}\right) + \alpha\!\!\!\sum_{S_{\vec{b}}^+\in\mathcal{S}_{S}^+}\!\!\Prob\left[S_{\vec{b}}^+|S,\vec{\kappa}\right]\tilde{J}(S_{\vec{b}}^+)\right\rbrace,
\label{eqn:vi_update}
\end{equation}
where $\Prob\left[S_{\vec{b}}^+|S,\vec{\kappa}\right]$ is the transition probability from state $S^+$ in one time step using $\vec{\kappa}$ given the transmission outcome vector $\vec{b}=\left[b_1,\dots,b_L\right]^\top, b_\ell\in\{0,1\}$, $\mathcal{S}^+_{S}$ denotes the set of all states reachable from $S$ with one time step transition\footnote{Intuitively, the set $\mathcal{S}_{S}^+$ corresponds to all possible transmission outcomes give the state $S$.}
$\mathcal{S}_{S}^+ \triangleq \left\lbrace S^+_{\vec{b}}=\Phi\left(S,\vec{b}\right)\,\forall\vec{b}\in\mathcal{B}\right\rbrace\label{eqn:reachable_state_space}$, $\mathcal{B}$ is the set of all possible transmission outcomes, i.e. all $L$-permutations of $\{0,1\}$, and $\tilde{J}(S_{\vec{b}}^+)$ denotes $L$-variate interpolation of the true value function $J(S_{\vec{b}}^+)$ on the discrete grid defined by $\mathcal{S}$.
\end{proposition}
\begin{proof}
From the Bellman optimality equation~\cite{Bertsekas_DP} we have
\begin{align}
J(S)&\triangleq\min_{\vec{\kappa}\in\bar{\Lambda}}\E\left[\rho\left(S,\vec{\kappa}\right) + \alpha J(S^+_{\vec{b}})\right]\nonumber\\
&\overset{(1)}{=}\min_{\vec{\kappa}\in\bar{\Lambda}}\rho\left(S,\vec{\kappa}\right) + \alpha \E\left[J(S^+_{\vec{b}})\right]\nonumber\\
&\overset{(2)}{=}\min_{\vec{\kappa}\in\bar{\Lambda}}\rho\left(S,\vec{\kappa}\right) + \alpha\!\!\!\sum_{S_{\vec{b}}^+\in\mathcal{S}_{S}^+}\!\!\Prob\left[S_{\vec{b}}^+|S,\vec{\kappa}\right]\tilde{J}(S_{\vec{b}}^+).
\end{align}
Where in step (1) we used the fact that $\rho\left(S,\vec{\kappa}\right)$ is deterministic, and in step (2) we expanded the expectation and sum over all possible next states defined by $\mathcal{S}^+_{S}$ and respective transition probabilities from $S$ to $S^+_{\vec{b}}$ (note that this set does not depend on $\vec{\kappa}$). 
Following our system model assumptions, transition probabilities are independent given $\vec{\kappa}$ and thus found as
\begin{align}
\Prob\left[S_{\vec{b}}^+|S,\vec{\kappa}\right] &= \Prob\left[\vec{\beta}=\vec{b}\,|\,\vec{\kappa}\right]  \notag\\
&= \sum_{\ell=1}^{L}\left(1-\kappa_{\ell}\right)\left(1-b_{\ell}\right) + \kappa_{\ell}b_{\ell}.
\label{eqn:trans_prob}
\end{align}
\end{proof}

The pseudo-code and the parameters of the resulting value iteration procedure are summarized, respectively, in Alg.~\ref{alg:avi} and Tab.~\ref{tab:list-of-algorithm-symbols}. 
Steps 3-7 of the algorithm compute a single pass over the state space to update $J(S)$. 
The update pass is repeated until convergence, defined in our algorithm by the threshold $\epsilon$.
Finally, once the optimal $J(S)$ is obtained, steps 9-11 retrieve the optimal policy based on the Bellman optimality equation
\begin{equation}
\vec{\kappa^\star}=
\argmin_{\vec{\kappa}\in\bar{\Lambda}}\left\lbrace\rho\left(S,\vec{\kappa}\right) + \alpha\hspace{-.2cm}\sum_{S_{\vec{b}}^+\in\mathcal{S}_{S}^+}\hspace{-.1cm}\Prob\left[S_{\vec{b}}^+|S,\vec{\kappa}\right]\tilde{J}(S_{\vec{b}}^+)\right\rbrace.
\label{eqn:opt_policy}
\end{equation}

 \begin{algorithm}
	\caption{Approximate Value Iteration Algorithm}
	\label{alg:avi}
	\begin{algorithmic}[1]
		\renewcommand{\algorithmicrequire}{\textbf{Input:}}
		\renewcommand{\algorithmicensure}{\textbf{Output:}}
		\REQUIRE States $\mathcal{S}$, actions $\Lambda_{\kappa}$, threshold $\epsilon$
		\ENSURE Optimal deterministic policy $\vec{p^\star}(S),\,\forall S\in\mathcal{S}$
		\\ \textit{Initialisation} : $J\left(S\right)\gets0,\forall S\in\mathcal{S}, \Delta \gets \epsilon$
		\WHILE {$\Delta > \epsilon$}
		\STATE $\Delta\gets0$
		\FOR {$S\in\mathcal{S}$}
		\STATE $J^{-}\gets J(S)$
		\STATE Update $J(S)$ according to~\eqref{eqn:vi_update}.
		\STATE $\Delta \gets \max\{\Delta,|J^{-}-J(S)|\}$
		\ENDFOR
		\ENDWHILE
		\FOR {$S\in\mathcal{S}$}
%		\COMMENT\STATE
		\STATE
Compute $\vec{\kappa^\star}$ using~\eqref{eqn:opt_policy}
\STATE $\vec{p^\star}(S)=\vec{\Psi}(\vec{\kappa^\star})$ 
\ENDFOR
		\RETURN $\vec{p^\star}(S),\,\forall S\in\mathcal{S}$
	\end{algorithmic} 
\end{algorithm}

In Fig.~\ref{fig:psr-schedules} the PSR policies are visualized for an exemplary case of two sensors ($L=2$), equal communication distances ($10$ m) and discretizations $M_\mathcal{S}=M_{\bar{\Lambda}}=30$. 
The left pair of plots shows the symmetric PSR policies of two homogeneous systems with $F_1=F_2=1.01$, while the left pair shows asymmetric PSR policies for heterogeneous systems with $F_1=1.01, F_2=1.2$.
\begin{remark}
	The discretization of the state space must take into account, additionally to the algorithm parameters $\lambda_\ell$ and $\alpha$, the covariance dynamics of the systems and the maximum transmission power.
	In fact, the selection of these parameters impacts the values of the target update function in~\eqref{eqn:vi_update}, leading to different types of policies.
	The misconfiguration of them could lead to approximation errors (for large discretization steps) or saturation (for low values of maximum estimation covariance) in the PSR policies.
\end{remark}

The complexity of the approximate value iteration depends on the cardinality of the system state space and action space. 
The state space $\mathcal{S}$ has cardinality $|\mathcal{S}| = M_{\mathcal{S}}^L$. 
The cardinality $M_{\bar{\Lambda}}$ of the action space depends on the discretized feasibility region of the power control problem discussed in~\ref{sec:tpc-problem}, and can vary depending on channel conditions and practical limitations on available power levels~\cite{802154-2006}.
The complexity of value iteration is thus $\mathcal{O}\left(|\mathcal{S}|^2\times|\bar{\Lambda}|\right)=\mathcal{O}\left(M_{\mathcal{S}}^{2L}\times M_{\bar{\Lambda}}\right)$. 

\begin{remark}
Discretization precision for state and action space $M_{\mathcal{S}},M_{\bar{\Lambda}}$, as well as the convergence threshold $\epsilon$, present a natural trade-off between optimality and complexity. 
In practice, they have to be carefully chosen based on the application. If better performance must be achieved and longer run-times are acceptable, higher precision and lower $\epsilon$ should be used. 
\end{remark}
\section{Evaluation}\label{sec:evaluation}
This section presents the remote state estimation performance and the transmission power expenditure of the network coordinated by the MDP problem of Sec.~\ref{sec:opt-problem}.
This is achieved by comprehensively evaluating the approximate algorithm of Sec.~\ref{sec:opt-algorithm} in all its parameters for different interference scenarios and dynamical systems.

The evaluation assumes the following parameters for wireless communication. 
Wireless sensors transmit on channel $26$ of the $2.4$ GHz ISM band, in an indoor propagation environment with parameters $\sigma^2=2.75$ dB and $\eta=3.3$~\cite{zamalloa2007analysis}, and absence of interference from other networks.
For every transmission, sensors select transmission powers $p_k \in \left[-24\text,7\right]$ dBm, and are subject to an AWGN with power $n_k=-100$ dBm. 
The parameters are selected based on the typical operating values of the IEEE Std. 802.15.4~\cite{802154-2006} RF SoC TI-CC2538 deployed in most recent experimental WSN platforms, e.g. Zolertia RE-Mote.
We evaluate different installations of sensors in a factory for two network topologies, \emph{circular} and \emph{assembly-line} shown in Fig.~\ref{fig:topologies} with varying parameters $d_1,d_2$ to investigate different interference scenarios.
The circular topology represents a scenario where the network infrastructure is scarce and multiple sensors (T) transmit to co-located receivers (R).
The assembly-line topology represents a dense factory environment where the sensors and their receivers are uniformly distributed.

We consider two classes of system dynamics.
Systems of class I are with $F_\ell^1=1.01$ and are more stable, while systems in class II are with $F_\ell^{2}=1.1$ and are less stable. 
Both systems are affected by the same system noise $R_{1,l}=0.4$ and their states are measured by the identical sensors with parameters $H_l=0.3$ and $R_{2,l}=1.1$. 
All measurements are transmitted to the respective remote state estimators in packets of $W=120$ bits.

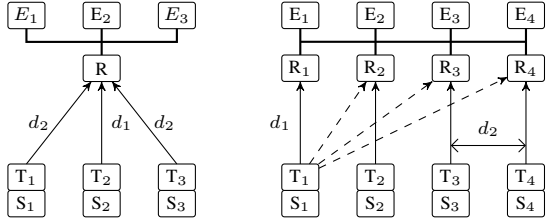
\begin{figure}
	\def\dx{1}
\def\dy{1.8}
\def\by{45}
\def\my{22.5}
\def\rad{15}
\def\sy{10}
\begin{tikzpicture}[
font = {\scriptsize},
tx_line/.style={line width= 0.4pt,->,>=stealth',dashed},
infra_line/.style={line width= 2pt,thick},
box/.style={draw,rectangle,minimum height=\sy pt,minimum width=14 pt,anchor=center,fill=white,inner sep=2pt,outer sep=0pt,rounded corners=1pt},
radiation/.style={{line width= 0.2pt,decorate,decoration={expanding waves,angle=35,segment length=3pt}}}]

\newcommand{\ncstx}[4]{
	\node[box] (p#3) at (#1,#2) {S$_#3$};
	\node[box, above=0pt of p#3] (t#3) {T$_#3$};}
	%\draw[radiation] ([xshift=0pt,yshift=3pt]t#3.north) -- ([xshift=#4pt,yshift=\rad-2 pt]t#3.north);}

\newcommand{\ncsrx}[3]{
	\node[box] (r#3) at (#1,#2) {R};
	%\draw[radiation] ([xshift=0pt,yshift=-3pt]r#3.south) -- ([xshift=0pt,yshift=-\rad pt]r#3.south);
	\node[box, above=\sy pt of r#3] (e#3) {E$_#3$};	
	\draw[infra_line] (r#3.north) to (e#3.south);}

\ncstx{0}{0}{1}{6}
\ncstx{\dx}{0}{2}{0}
\ncstx{2*\dx}{0}{3}{-6}

\ncsrx{\dx}{\dy}{2}

\coordinate (midway) at ([yshift=5pt]r2.north);
\coordinate (e1-net) at ([xshift=-\dx cm]midway);
\node[box] (e1) at ([xshift=-\dx cm]e2) {$E_1$};
\draw[infra_line] (midway) to (e1-net) to (e1);
\coordinate (e3-net) at ([xshift=\dx cm]midway);
\node[box] (e3) at ([xshift=\dx cm]e2) {$E_3$};
\draw[infra_line] (midway) to (e3-net) to (e3);

\draw[tx_line,solid](t1.north) to node[left] {$d_2$} (r2);
\draw[tx_line,solid](t2.north) to node[right] {$d_1$} (r2);
\draw[tx_line,solid](t3.north) to node[right] {$d_2$} (r2);

\end{tikzpicture}
\def\dx{1}
\def\dy{1.45}
\def\by{45}
\def\my{22.5}
\def\rad{15}
\def\sy{10}
\begin{tikzpicture}[
font = {\scriptsize},
tx_line/.style={line width= 0.4pt,->,>=stealth',dashed},
infra_line/.style={line width= 2pt,thick},
box/.style={draw,rectangle,minimum height=\sy pt,minimum width=14 pt,anchor=center,fill=white,inner sep=2pt,outer sep=0pt,rounded corners=1pt},
radiation/.style={{line width= 0.2pt,decorate,decoration={expanding waves,angle=35,segment length=3pt}}}]

\newcommand{\ncs}[3]{
\node[box] (p#3) at (#1,#2) {S$_#3$};
\node[box, above=0pt of p#3] (t#3) {T$_#3$};
%\draw[radiation] ([xshift=0pt,yshift=3pt]t#3.north) -- ([xshift=0pt,yshift=\rad pt]t#3.north);
\node[box] (r#3) at ([yshift=\dy cm]t#3) {R$_#3$};
%\draw[radiation] ([xshift=0pt,yshift=-3pt]r#3.south) -- ([xshift=0pt,yshift=-\rad pt]r#3.south);
\node[box, above=\sy pt of r#3] (e#3) {E$_#3$};
\draw[infra_line] (r#3.north) to (e#3.south);}

\ncs{0}{0}{1}
\ncs{1*\dx}{0}{2}
\ncs{2*\dx}{0}{3}
\ncs{3*\dx}{0}{4}

\draw[tx_line,solid](t1) to node[left] {$d_1$} (r1);
\draw[tx_line](t1) to (r2);
\draw[tx_line](t1) to (r3);
\draw[tx_line](t1) to (r4);
\draw[tx_line,solid](t2) to (r2);
\draw[tx_line,solid](t3) to (r3);
\draw[tx_line,solid](t4) to (r4);

\draw[infra_line] ([xshift=0pt,yshift=5pt]r1.north) to ([xshift=0pt,yshift=5pt]r4.north);

\draw[tx_line,angle 90-angle 90,solid] ([yshift=7pt]t3.north) to node[above] {$d_2$} ([yshift=7pt]t4.north);

%\node (dots) at (2*\dx,\sy pt) {$\cdots$};
%\node (coord-net) at ([xshift=\dx cm,yshift=2pt]r2.north) {};
%\node[box] (coord) at ([yshift=25 pt]coord-net) {NM};
%\draw[infra_line] (coord-net) to (coord);

\end{tikzpicture}%\\
	\caption{Circular (left) and assembly-line (right) topologies representing resp. scarce and dense network infrastructure deployments. Multiple sensors  communicate (T, R) system's (S) measurements to remote state estimators (E) with varying parameters $d_1$, $d_2$.}
	\label{fig:topologies}
\end{figure}

% Multisim 3 users
% N = 8;d1 = 5;
% F = 1.01*ones(L,1); H = 0.30; R1 = 0.4; R2 = 1.1; B = 1;
% Psize = 10; Pmax = 20;
% epsilon = 5e-2;
We started by presenting an exhaustive evaluation of the algorithm's sensitivity to design parameters $\alpha$ and $\lambda_\ell=\lambda$ for different interference scenario.
We achieved this by performing Monte Carlo simulations for a network
of three systems of class I in a circular topology (Fig.~\ref{fig:topologies}), and by computing, for every scenario, the total mean network estimation error covariance $\sum_l \bar{P}_\ell$ and transmission power $\sum_l \bar{p}_\ell$.
The two metrics indicate, respectively, the remote state estimation performance and energy consumption of the network.
In each scenario, we calculated the optimal transmission power policies from an action space with $M_{\bar{\Lambda}}=512$ values in the interval $\left(0,1\right)$ %quantizing the  transmission powers over $8$ values, 
and a state space with $M_{\mathcal{S}}=10$ in the interval $\left[0,20\right]$. 
The algorithm precision is selected as $\epsilon=0.05$.

\begin{figure}[t!]
	\includegraphics{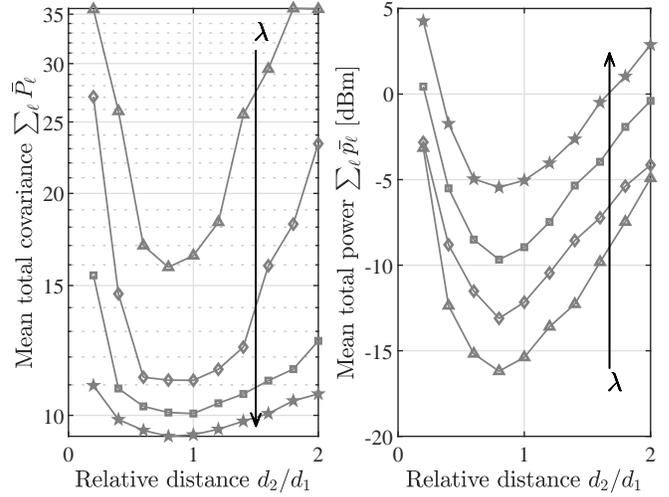}
	\caption{Impact of the relative distance $d_2/d_1$, $d_1=10$ m on the system's performance for 3 sensors in a circular topology for increasing values of $\lambda$ (arrow) and $\alpha=1$.}
	\label{fig:multisim-distances}
\end{figure}

Fig.~\ref{fig:multisim-distances} shows the mean network covariance and transmission power for different relative distances $d_2/d_1$, $d_1=10$ m and $\alpha=1$.
In this configuration of the circular topology (Fig.~\ref{fig:topologies}), all the receivers are placed at the same location, while transmitter 2 is placed at distance $d_1$ and transmitters 1 and 3 at distance $d_2$. 
We observe that, as the difference between sensors' positions decreases, i.e. $d_2/d_1\sim1$, the mean values of the total estimation error covariance and transmission power decrease, reaching their minimum values.
This result shows that mutual interference has a strong influence on the performance of the system, requiring more power when the interference levels are unevenly distributed in the network.
As expected, the parameter $\lambda$, allowing to trade-off the total estimation error over the transmission power expenditure, plays an important role in determining the total system performance and can be set according to use case requirements.
In fact, when $\lambda$ increases, the mean transmission powers increase, decreasing the mean estimation error covariances.

\begin{figure}[t!]
	\includegraphics{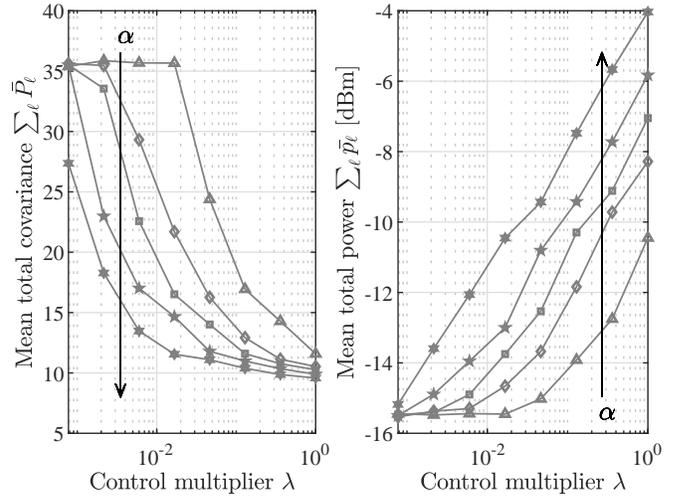}
	\caption{Impact of the trade-off parameter $\lambda$ on the system's performance for 3 sensors in a circular topology with $d_2/d_1=1.2$ and for increasing values of $\alpha$ (arrow).}
	\label{fig:multisim-lambda}
\end{figure}

\begin{figure}[t!]
	\includegraphics{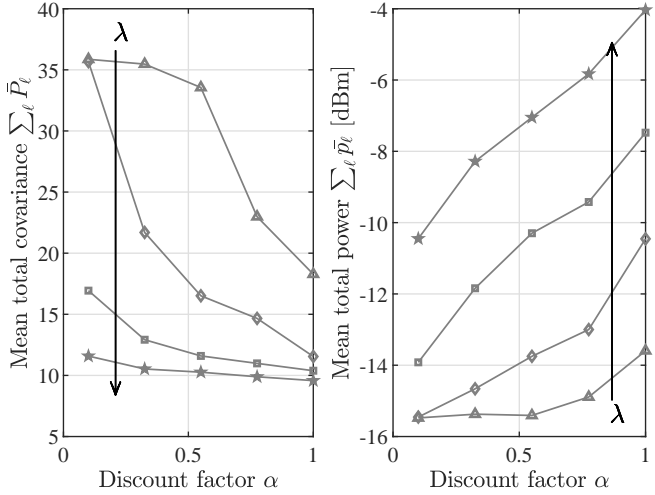}
	\caption{Impact of the discount factor $\alpha$ on the system's performance for 3 sensors in a circular topology with $d_2/d_1=1.2$ and for increasing values of $\lambda$ (arrow).}
	\label{fig:multisim-alpha}
\end{figure}

Fig.~\ref{fig:multisim-lambda} and Fig.~\ref{fig:multisim-alpha} show the interplay between $\lambda$ and $\alpha$ for the fixed distances $d_1=10$ m and $d_2=12$ m.
In all scenarios, as expected, by increasing $\lambda$ and $\alpha$, the transmission powers increase, resulting in lower covariances.
This effect, however, is caused by different reasons.
While $\lambda$ regulates the trade-off between instantaneous values of transmission power and estimation error covariance, $\alpha$ weights the importance of future system performance with respect to current values, where $\alpha\rightarrow1$ indicates more importance of future values.
As shown in Fig.~\ref{fig:multisim-alpha}, lower values of $\alpha$ lead to policies that operate over longer horizons by means of lower transmission powers at the cost of higher estimation error covariances.
From these results, we can conclude that the proposed approximated algorithm flexibly adapts to different network configurations, and arbitrary desired system performance can be achieved with the accurate selection of $\alpha$ and $\lambda$.

 % Time evolution 4 users different distances 1000,350
% N=8; d1 = 10; d2 = 0.35*d1; position_pattern='assembly-line';
% F = 1.01*ones(L,1); H = 0.30; R1 = 0.4; R2 = 1.1; B = 1;
% Psize: 8;  Pmax: 10;
% alpha: 0.9000
% multiplier: 0.0100
% epsilon: 0.0500

\begin{figure}[!t]
	\includegraphics{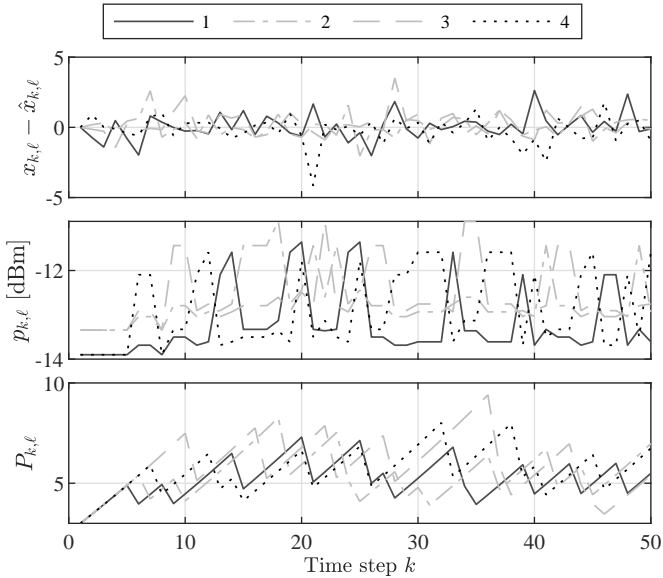}
	\caption{Time evolution of the estimation error (top), transmission power (middle), and estimation error covariance (bottom) for an assembly-line topology of identical 4 systems $d_1=10$ m, $d_2=3.5$ m.}%, edge links 1,4 (darker) show better performance than core links 1,3 (lighter).}
	\label{fig:simulation_4_homo_users_asymm}
\end{figure}

We continued the evaluation of the algorithm by presenting trajectories of transmission powers, estimation errors, and covariances for two exemplary scenarios.
Fig.~\ref{fig:simulation_4_homo_users_asymm} shows a network of four systems of class I in an assembly-line topology (Fig.~\ref{fig:topologies}) with distances $d_1=10$ m and $d_2=3.5$ m.
The optimal policies are obtained from an action space with $M_{\bar{\Lambda}}=512$ values in the interval $\left(0,1\right)$ %quantizing the  transmission powers over $8$ values, 
and a state space with $M_{\mathcal{S}}=8$ in the interval $\left[0,10\right]$. 
The algorithm parameters are $\epsilon=0.05$, $\lambda=0.01$, and $\alpha=0.9$.
From Fig.~\ref{fig:simulation_4_homo_users_asymm}, we can observe that communication links 1,4, placed at the edge of the topology, experience less interference and use lower transmission powers, while internal links 2,3 use higher transmission powers.
As all systems belong to the same class, they all present a similar evolution of the estimation errors and covariances, with slightly higher values for more interfered links.
Also, in this case, we can observe that the interference levels strongly influence the selection of transmission powers and the trajectories of the systems.

% Time evolution 3 users different classes, 2 class II
% N=20; d1 = 10; d2 = 1*d1; position_pattern='double-circular';
% F = [ 1.0100, 1.1000, 1.1000]; H = 0.30; R1 = 0.4; R2 = 1.1; B = 1;
% Psize: 20; Pmax: 30
% alpha: 0.9000
% multiplier: 0.0100
% epsilon: 0.0500
\begin{figure}[!t]
	\includegraphics{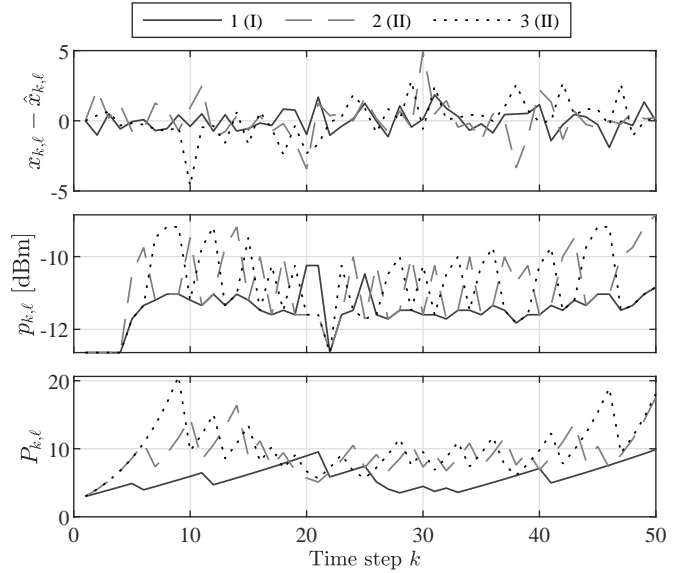}
	\caption{Time evolution of the estimation error (top), transmission power (middle), and estimation error covariance (bottom) for a circular topology of 3 systems, 2 of class II and 1 of class I, and equal distances $d_1,d_2=10$ m.}
	\label{fig:simulation_3_hete_users_symm}
\end{figure}

On the other hand, Fig.~\ref{fig:simulation_3_hete_users_symm} shows the trajectories of a network of heterogeneous systems, one from class I and two from class II.
In this scenario, all nodes have the same communication distance $d_1,d_2=10$ m according to the circular topology (Fig.~\ref{fig:topologies}).
The optimal power policies are obtained from an action space with $M_{\bar{\Lambda}}=8\cdot 10^3$ in the interval $\left(0,1\right)$ %quantizing the  transmission powers over $8$ values, 
and a state space with $M_{\mathcal{S}}=20$ in the interval $\left[0,30\right]$. 
The algorithm parameters are $\epsilon=0.05$, $\lambda=0.01$, and $\alpha=0.9$.
In this scenario, the difference between trajectories is higher, leading to higher transmission powers, estimation errors, and covariances for systems of class II and lower for the system of class I.
%In this case, the estimation errors and covariances of the systems are different, experiencing higher values and deviations for systems of class II.
From all the evaluation scenarios we conclude that the transmit power policies can correctly support different types of system dynamics and topologies, capturing the properties of the systems and optimally distributing network resources accordingly.

\section{Conclusions}\label{sec:conclusions}
In this work, we developed a novel optimal transmission power control policy for industrial WSNs deployed to transmit measurements of multiple independent NCSs' dynamics to remote state estimators.
The policy centrally coordinates the simultaneous access to the shared communication channel by adjusting the transmission powers of the sensors,
and it is obtained by formulating an infinite horizon MDP optimization problem that combines the network's transmission powers and estimation error covariances.
We show that the MDP problem can be solved by independently optimizing the network Packet Success Rates (PSRs) and the minimum transmission powers that achieve them. 
We propose an approximate value iteration algorithm for its implementation in practical scenarios.

Furthermore, we performed an exhaustive evaluation of the algorithm, proving its effectiveness to adapt to different interference scenarios and system's dynamics.
The evaluation results provide a comprehensive characterization of the algorithm's parameters for arbitrary estimation performances and transmission power expenditures, proving that, by varying its main trade-off parameters $\lambda$ and $\alpha$, it is possible to flexibly adapt the algorithm's performance to arbitrary use cases.
Further developments are possible and can investigate distributed implementations or time-varying wireless channels.

% use section* for acknowledgment
\section*{Acknowledgment}
This work has been carried out with the support of the Technical University of Munich - Institute for Advanced Study, funded by the German Excellence Initiative and the DFG Priority Program SPP1914 ``Cyber-Physical Networking'' grant number KE1863/5-1.

% References
\bibliographystyle{IEEEtran}
\bibliography{references,powercontrol}

\end{document}